\documentclass{article}
\usepackage{ijcai21}

\usepackage{times}

\usepackage{soul}
\usepackage{url}
\usepackage[hidelinks]{hyperref}
\usepackage[utf8]{inputenc}
\usepackage[small]{caption}
\usepackage{graphicx}
\usepackage{amsmath}
\usepackage{booktabs}
\usepackage[T1]{fontenc}    
\usepackage{booktabs}       
\usepackage{amsfonts}       
\usepackage{nicefrac}       
\usepackage{microtype}      
\usepackage{xcolor}         

\usepackage{algorithm}
\usepackage{algorithmic}
\usepackage{multirow}
\usepackage{color}
\usepackage{rotating}

\usepackage{amsthm}

\newtheorem{proposition}{Proposition}

\newenvironment{proof*}[1]
  {\renewcommand{\proofname}{#1}%
   \begin{proof}}
  {\end{proof}}

\newcommand{\reals}{\mathbb{R}}

\newcommand{\wVec}{\mathbf{w}}
\newcommand{\xVec}{\mathbf{x}}

\newcommand{\ones}{\mathbf{1}}

\title{Optimal Edge Weight Perturbations to Attack Shortest Paths}

\author{
Benjamin A. Miller$^1$\and
Zohair Shafi$^1$\and
Wheeler Ruml$^2$\and
Yevgeniy Vorobeychik$^3$\and\\
Tina Eliassi-Rad$^1$\And
Scott Alfeld$^4$\\
\affiliations
$^1$Northeastern University, Boston, MA, USA\\
$^2$University of New Hampshire, Durham, NH, USA\\
$^3$Washington University in St. Louis, St. Louis, MO, USA\\
$^4$Amherst College, Amherst, MA, USA\\
\emails
\{miller.be, shafi.z, t.eliassirad\}@northeastern.edu,
ruml@cs.unh.edu,
yvorobeychik@wustl.edu,
salfeld@amherst.edu
}

\begin{document}

\maketitle

\begin{abstract}

Finding shortest paths in a given network (e.g., a computer network or a road network) is a well-studied task with many applications. We consider this task under the presence of an adversary, who can manipulate the network by perturbing its edge weights to gain an advantage over others. Specifically, we introduce the \emph{Force Path Problem} as follows. Given a network, the adversary's goal is to make a specific path the shortest by adding weights to edges in the network.  The version of this problem in which the adversary can cut edges is NP-complete. However, we show that Force Path can be solved to within arbitrary numerical precision in polynomial time. We propose the \texttt{PATHPERTURB} algorithm, which uses constraint generation to build a set of constraints that require paths other than the adversary's target to be sufficiently long. Across a highly varied set of synthetic and real networks, we show that the optimal solution often reduces the required perturbation budget by about half when compared to a greedy baseline method.

\end{abstract}

\section{Introduction}
\label{sec:intro}
The shortest path problem is a seminal task in graph theory with numerous real-world applications in computer networks, transportation networks, etc. Given two nodes in a network, the shortest path between them is the set of edges that connects the two nodes with the minimum sum of edge weights. For a given network, two nodes can have more than one shortest path; and the network can be directed or undirected. Here we only consider undirected networks.

In this paper, we present the \emph{Force Path Problem}, where there is a specific path that the adversary wants to be the shortest path between a pair of source and destination nodes. The adversary can increase weights of edges and has a fixed budget with which to achieve this goal. The \emph{Force Path Problem} is similar to the \emph{Force Path Cut Problem}~\cite{Miller2021}. The difference is in the attack vector: in this case the adversary makes edges more expensive (by increasing edge weights) rather than removing edges. This difference may seem relatively minor, but it has a profound implication for the computational complexity of the problem. While Force Path Cut is NP-complete, we show in this paper that Force Path can be solved within arbitrary precision in polynomial time. We demonstrate that Force Path can be formulated as a linear program with a constraint set that is potentially factorial in the number of nodes. There is, however, a natural polynomial-time oracle to find violated constraints in any candidate solution (which we use to include a subset of constraints as necessary). We propose the \texttt{PATHPERTURB} algorithm that uses the oracle to iteratively refine the graph perturbations until the target path is the shortest.\footnote{We use the terms graph and network interchangeably.}

The main contributions of the paper are as follows:
\begin{itemize}
    \item We formally define the Force Path problem: an adversarial attack on shortest paths. 
    \item We formulate an oracle to identify the most violated constraint at any given point, the existence of which implies that Force Path can be optimized within arbitrary precision in polynomial time.
    \item We propose the \texttt{PATHPERTURB} algorithm, which uses the oracle to minimize the required perturbation budget. 
    \item We present the results of experiments on synthetic and real networks, in which \texttt{PATHPERTURB} reliably reduces the required perturbation budget compared to a greedy baseline method.
\end{itemize}

\section{Force Path Problem Definition}
\label{sec:model}

Consider a graph $G=(V, E)$, with a set of nodes $V$ and edges $E$, where $|V|=N$ and $|E|=M$. The edges in $E$ are undirected and have nonnegative weights $w:E\rightarrow\reals_{\geq0}$. The edge weights denote distances (a.k.a.~lengths) between the adjacent nodes. 
In addition to the weighted graph, we are given a pair of source and destination nodes $s,t\in V$. The adversary's goal is to make a specific path, $p^*$, the shortest path from $s$ to $t$ in $G$. The adversary can achieve this by arbitrarily increasing the weight of any edge in $G$, all of which are visible to him/her. Within a budget constraint $b$, the adversary increases edge weights to obtain new weights $w^\prime$ such that $\sum_{e\in E}{\left(w^\prime(e)-w(e)\right)}\leq b$ and $p^*$ is the (possibly exclusive) shortest path from $s$ to $t$.

In the next section, we will show that the Force Path Problem can be formulated as a linear program (LP), which implies a polynomial time algorithm to get a solution within any specified precision, despite a very large number of constraints.

\section{Force Path LP Formulation}
\label{sec:formulation}
Let $\wVec\in\reals_{\geq0}^M$ be a vector of edge weights in the original graph $G$ and $\Delta\in\reals_{\geq0}^M$ be a vector of edge-weight perturbations. 
 For any path $p$ from $s$ to $t$ in $G$, let $\xVec_p\in\{0, 1\}^M$ be an edge indicator vector for $p$: the entries for $\xVec_p$ associated with the edges that comprise $p$ are $1$, while all other entries are $0$. Thus, $\wVec^\top \xVec_p$ is the length of $p$ in the original graph and $(\wVec+\Delta)^\top \xVec_p$ is its length in the perturbed graph. 

The linear program formulation of Force Path is based on two key observations. First, any path $p$ that is not longer than $p^*$ must be perturbed to be longer than $p^*$. This is clear from the fact that, if it were not the case, $p^*$ would not be the shortest path. Second, we do not perturb $p^*$, formalized as follows:
\begin{proposition}
The minimum-budget solution to Force Path includes no perturbation of any edge along $p^*$.
\end{proposition}
\begin{proof}
Let $\hat{\Delta}$ be the minimum-budget solution. Suppose that the claim is not true---i.e., there is an edge $e$, which is part of $p^*$, where $\hat{\Delta}(e) > 0$. Since $\hat{\Delta}$ is the solution to Force Path, $p^*$ is the shortest path when $\hat{\Delta}$ is added to the edge weights. Let $\Delta^\prime$ be the same vector with the entry at $e$ reduced to $0$, i.e.,
$\Delta^\prime(e^\prime)=\hat{\Delta}(e^\prime)$ for $e^\prime\in E\setminus \{e\}$ and $\Delta^\prime(e)=0$. Since $e$ is part of $p^*$, $\xVec_{p^*}^\top\hat{\Delta}-\xVec_{p^*}^\top\Delta^\prime=\hat{\Delta}(e)$. For any path $p$, since $\xVec_p$ consists of only zeros and ones, we have
\begin{equation}
    \xVec_{p}^\top\hat{\Delta}-\xVec_{p}^\top\Delta^\prime\leq\hat{\Delta}(e).\label{eq:pathReduction}
\end{equation}
Again, $p^*$ must be the shortest path using $\hat{\Delta}$, and, therefore, for any path $p$ from $s$ to $t$, $\xVec_{p}^\top\hat{\Delta}\geq\xVec_{p^*}^\top\Delta^\prime$. Combining this with (\ref{eq:pathReduction}), we have, for any path $p$
\begin{equation}
    \xVec_{p^*}^\top\Delta^\prime=\xVec_{p^*}^\top\hat{\Delta}-\hat{\Delta}(e)\leq\xVec_{p^*}^\top\hat{\Delta}-\hat{\Delta}(e)\leq\xVec_p^\top\Delta^\prime.
\end{equation}
 Thus, if $\hat{\Delta}$ were replaced with $\Delta^\prime$, $p^*$ would still be the shortest path. This implies that the total perturbation could be reduced by $\hat{\Delta}(e)>0$ and still achieve the objective, so $\hat{\Delta}$ is not the minimum-budget solution. Thus, we have a contradiction, and the proof is complete.
\end{proof}

The implication of this observation is that there is a fixed lower bound for the lengths of all paths. Let $\ell=\wVec^\top\xVec_{p^*}$ be the length of $p^*$ and $P_{\ell}$ be the set of paths from $s$ to $t$ whose length is less than or equal to $\ell$. Finally, let $\delta$ be the ``buffer'' we use to ensure $p^*$ is the unique shortest path: the difference between the length of $p^*$ and the length of the second shortest path\footnote{In a scenario where being tied for shortest is acceptable, $\delta$ can be set to $0$. If it is acceptable for $p^*$ to be shortest by any $\epsilon>0$, we can set $\delta$ to $0$ and distribute an arbitrarily small value across edges not on $p^*$. In this case, the budget must be strictly larger than the sum of the computed perturbations.}. We formulate the linear program as follows:
\begin{align}
  \hat{\Delta} = &\arg\min_{\Delta} \ones^\top\Delta\label{eq:minCost}\\
  \text{s.t.} &~\Delta_i\geq 0,~ ~1\leq i\leq M\label{eq:posPert}\\
  &~\left(\wVec+\Delta\right)^\top\xVec_p\geq \ell+\delta,~ ~\forall p\in P_{\ell+\delta}\setminus \{p^*\}\label{eq:tooShort}\\
  &~\xVec_{p^*}^\top\Delta=0\label{eq:dontChangePStar}.
\end{align}

As with the approximate version of Force Path Cut discussed in~\cite{Miller2021}, the number of paths in $P_{\ell+\delta}$ may be too large to enumerate all constraints. In an $N$-node clique, for example, the number of paths of length $N-1$ between any two nodes is $(N-2)!$. Thus, specifying all constraints in the linear program is computationally intractable. We use constraint generation to iteratively incorporate constraints as they are needed (see, e.g.,~\cite{Ben-Ameur2006,Letchford2013}). In order to use constraint generation, however, there must be an oracle that returns a constraint being violated at a given point.

There is, fortunately, a natural oracle for the constraints specified in~(\ref{eq:tooShort}), which not only returns a violated constraint, but the constraint \emph{most} violated at the proposed solution. We find this constraint as follows. The candidate solution is a perturbation to the edge weights, $\hat{\Delta}$. We apply the perturbation to get the new edge weights $w^\prime$, where
\begin{equation}
    w^\prime(e)=w(e)+\hat{\Delta}(e).\label{eq:updateWeight}
\end{equation}
Using $w^\prime$ as distances, we find the shortest path $p$ from $s$ to $t$ in $G$. If $p$ is $p^*$, we find the second shortest path if it exists. If there is no such path, there is no violated constraint. If there is, we let $p$ be this path. If $p$ is at least $\delta$ longer than $p^*$, there is no violated constraint. If not, the length of $p$ needs to be incorporated as a constraint.  Algorithm~\ref{alg:oracle} provides the pseudocode for this procedure.

\begin{algorithm}[tb]
\caption{ConstraintOracle}
\label{alg:oracle}
\textbf{Input}: graph $G$, weights $w$, target path $p^*$, buffer $\delta$\\
\textbf{Output}: A path $p$ from $s$ to $t$ in $G$
\begin{algorithmic}[1] 
\STATE $s\gets$ first node in $p^*$
\STATE $t\gets$ last node in $p^*$
\STATE $p\gets$ shortest path from $s$ to $t$ in $G$
\IF{$p$ is $p^*$}
\STATE $p\gets$ second shortest path from $s$ to $t$ in $G$
\STATE $\langle\langle p$ will be $\emptyset$ with length $\infty$ if $p^*$ is the only path$\rangle\rangle$
\ENDIF
\IF{length$(p)\geq$ length$(p^*)+\delta$}
\STATE $p\gets\emptyset$
\ENDIF
\RETURN $p$
\end{algorithmic}
\end{algorithm}

While the number of constraints may be extremely large, each one is a standard linear inequality constraint, which implies that the feasible region is convex. Since we have a constraint oracle that runs in polynomial time,\footnote{Finding the two shortest simple paths between two nodes takes $O(NM)$ time using Yen's algorithm~\cite{Yen1971}. If $\delta=0$ and edge weights are strictly positive, we can use an algorithm not restricted to simple paths that runs in $O(M+N\log{N})$ time~\cite{Eppstein1998}.} this system can be optimized in polynomial time regardless of the number of constraints. Using the ellipsoid algorithm introduced by Khachiyan (see~\cite{Gacs1981}), we can solve a linear program within finite precision in a polynomial number of iterations~\cite{Grotschel1981}. This results in the following proposition.

\begin{proposition}
Force Path can be optimized within precision of any constant $\epsilon>0$ in polynomial time.
\end{proposition}

\section{Proposed Method: PATHPERTURB}
\label{sec:pathattack}
While the ellipsoid algorithm provably converges in polynomial time, it is considerably slower in practice than simplex methods. Thus, our proposed algorithm iteratively solves a linear optimization procedure, adding constraints via the oracle as necessary. We call this algorithm \texttt{PATHPERTURB}.

Our perturbation algorithm uses a linear program where each constraint is associated with a path from $s$ to $t$ that is not longer than $p^*$. Each path must have weights added to the edges so that the path's length will become sufficiently long. \texttt{PATHPERTURB} operates in an iterative fashion as it builds the set of necessary constraints. At each iteration, it finds a solution based on a subset of constraints, perturbs the weights based on this solution, and, if there is still a path from $s$ to $t$ that is shorter, it adds the corresponding constraint to the linear program. Algorithm~\ref{alg:PATHPERTURB} provides \texttt{PATHPERTURB}'s pseudocode.
\begin{algorithm}[tb]
\caption{PATHPERTURB}
\label{alg:PATHPERTURB}
\textbf{Input}: graph $G=(V, E)$, weights $w$, target path $p^*$, buffer $\delta$\\
\textbf{Output}: perturbation vector $\hat{\Delta}$
\begin{algorithmic}[1]
\STATE $\ell\gets$ length of $p^*$
\STATE $\wVec\gets$ weight vector for $w$
\STATE $P_{\ell+\delta}\gets\emptyset$
\STATE $M\gets|E|$
\STATE $p\gets p^*$
\REPEAT
\STATE $P_{\ell+\delta}\gets P_{\ell+\delta}\cup \{p\}$
\STATE $\hat{\Delta}\gets$ solution to (\ref{eq:minCost})--(\ref{eq:dontChangePStar})
\STATE $w^\prime\gets w+\hat{\Delta}$ ~ ~ $\langle\langle$as in (\ref{eq:updateWeight})$\rangle\rangle$
\STATE $p\gets$ ConstraintOracle($G$, $w^\prime$, $p^*$, $\delta$)
\UNTIL{$p$ is empty}
\RETURN $\hat{\Delta}$
\end{algorithmic}
\end{algorithm}

\section{Experiments}
\label{sec:setup}

This section presents the baseline methods, the networks used in experiments, the experimental setup, and the results.

\subsection{Baseline Methods}
We compare \texttt{PATHPERTURB} to two simple greedy baseline algorithms. Each algorithm iteratively perturbs a single edge on the shortest path $p$ from $s$ to $t$ until $p^*$ is the shortest path (and, if the buffer $\delta$ is greater than 0, until the second shortest path is at least $\delta$ longer than $p^*$). The first baseline we consider, \texttt{GreedyFirst}, perturbs the first edge (in path traversal order) in $p$ that deviates from $p^*$.  We also use a method in which we perturb the edge with the smallest weight of all edges that are in $p$ but not $p^*$. We refer to this baseline as \texttt{GreedyMin}. In both cases, the selected edge is perturbed enough to make the path at least $\delta$ longer than $p^*$---i.e., if $\xVec_p$ is the edge indicator vector for the current shortest path and $\ell^\prime=(\wVec+\Delta)^\top \xVec_p$, the entry in $\Delta$ associated with the selected edge is increased by $\ell+\delta-\ell^\prime$.

\subsection{Synthetic and Real Networks}
\label{sec:graphdata}
We ran \texttt{PATHPERTURB} and the baseline algorithms on both synthetic and real networks. All networks are undirected. 

Our synthetic networks span a wide variety of topologies:
\begin{itemize}
    \item Erd\H{o}s--R\'{e}nyi (ER) random networks with 16,000 nodes and an edge probability of 0.00125
    \item Barab\'{a}si--Albert (BA) graphs with 16,000 nodes, where each new node connects to 10 existing nodes
    \item Watts--Strogatz (WS) graphs with 16,000 nodes, average degree $20$, and edge rewiring probability $0.02$
    \item Stochastic Kronecker (KR) graphs with $2^{14}$ nodes and a density parameter of $0.0125$
    \item $285\times285$ two-dimensional lattice (LAT) networks
    \item $565$-vertex complete (COMP) graphs
\end{itemize}
In all cases, we generate 100 networks from random (or fixed) network generators and add edge weights. Note that the number of edges is approximately 160,000 in all synthetic networks. We also consider various edge-weight distribution for each synthetic network:
\begin{itemize}
    \item Option 1: Give all edges weight $1$.
    \item Option 2: For each edge, draw a value from a Poisson distribution with rate parameter 20, and add $1$ to get the weight.
    \item Option 3: Draw each weight from a uniform distribution over integers from 1 to 41.
\end{itemize}

In addition to synthetic graphs, we use the following real networks:
\begin{itemize}
    \item Oregon autonomous systems (AS)~\cite{Leskovec2005}
    \item Wikispeedia (WIKI)~\cite{West2009}
    \item Pennsylvania roads (PA-ROAD)~\cite{Leskovec2009}
    \item Northeast US roads (NEUS)\footnote{Available at \url{https://bit.ly/2QWcug9}.}
    \item Central Chilean power grid (GRID)~\cite{Kim2018}
    \item Lawrence Berkeley National Laboratory computer network traffic (LBL)\footnote{Available at \url{https://bit.ly/2PQbOsr}.}
    \item DBLP coauthorship graph (DBLP)~\cite{Benson2018}
\end{itemize}
Networks AS, WIKI, and PA-ROAD do not have weights on their edges, so we add weights similar to the synthetic networks. In LBL and DBLP, the weights represent similarities rather than distances---number of connections and number of coauthored papers, respectively---so we invert the weight for use in the shortest path computation. In these cases, we set $\delta$ to 0.1, while we use $\delta=1$ in all other cases.

\subsection{Experimental Setup}
We run 100 trials for each network (or network generator) and each weighting scheme if applicable. In each experiment, we select $s$ and $t$ uniformly at random from the largest connected component of $G$, with the exception of LAT, PA-ROAD, and NEUS. In these grid-like graphs, computing the sequence of shortest simple paths is extremely time consuming, and we instead select $s$ at random and choose $t$ from among the nodes $50$ hops away from $s$. We then compute the 100th, 200th, 400th, and 800th shortest paths from $s$ to $t$ and use these as $p^*$. For LAT, PA-ROAD, and NEUS, we compute these paths only using the induced subgraph of nodes that are at most 60 hops away from $s$.

We ran the experiments using a Linux cluster with 32 cores and 192 GB of memory per node. We implemented the linear program in \texttt{PATHPERTURB} using the Python interface to Gurobi 9.1.1, and sequential shortest paths were computed using  \texttt{shortest\_simple\_paths} in NetworkX.\footnote{Gurobi is available at \url{https://www.gurobi.com}. NetworkX is available at \url{https://networkx.org}.}

\subsection{Results}
\label{sec:results}
\begin{figure*}
    \centering
    \includegraphics[width=\textwidth]{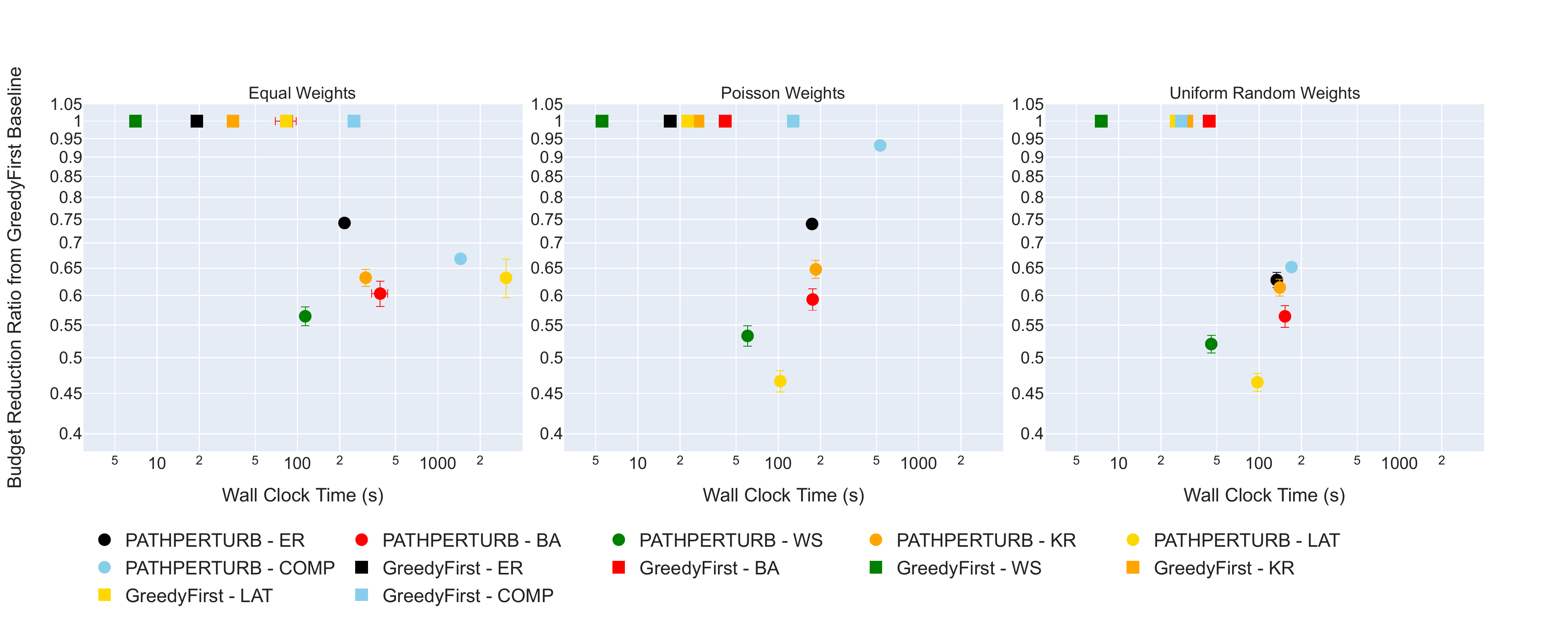}
    \caption{Results on synthetic networks. Results are shown using \texttt{PATHPERTURB} ($\circ$) and \texttt{GreedyFirst} ($\square$), and each color represents a different network. The plots present results when the weights are equal (left), when they are drawn from a Poisson distribution (center), and drawn from a uniform distribution (right). Each plot shows the required budget as a proportion of the budget required using \texttt{GreedyFirst} (vertical axis) with respect to wall clock running time (horizontal axis). Lower cost reduction ratio and lower wall clock time (toward the lower left) is better.  Error bars represent standard errors. In nearly all cases, \texttt{PATHPERTURB} yields a substantial cost reduction for its additional running time, though cliques with Poisson weights are nearly optimized with the baseline. Note: the black square (\texttt{GreedyFirst} on ER) in the right-hand plot is obscured by the yellow square.}
    \label{fig:plots_sim_gs}
\end{figure*}

We treat the result of \texttt{GreedyFirst} as our baseline budget and report the value optimized by \texttt{PATHPERTURB} as a reduction from the baseline. With few exceptions, \texttt{GreedyFirst} outperforms \texttt{GreedyMin} in both running time and perturbation cost, so we omit the \texttt{GreedyMin} results for clarity of presentation. For each graph, we use the algorithms in an attempt to minimize the budget, after which the adversary would determine whether or not the attack is possible within the constraints. Figure~\ref{fig:plots_sim_gs} shows the results on the synthetic networks, while Figure~\ref{fig:plots_ruw_gs} shows the results on real networks with synthetic edge weights; and
Figure~\ref{fig:plots_rw_gs} shows the results on real weighted networks. The figures show the results where $p^*$ is the 800th shortest path. Due to space limitations, we omit the other results; they are substantially similar.

\begin{figure*}
    \centering
    \includegraphics[width=\textwidth]{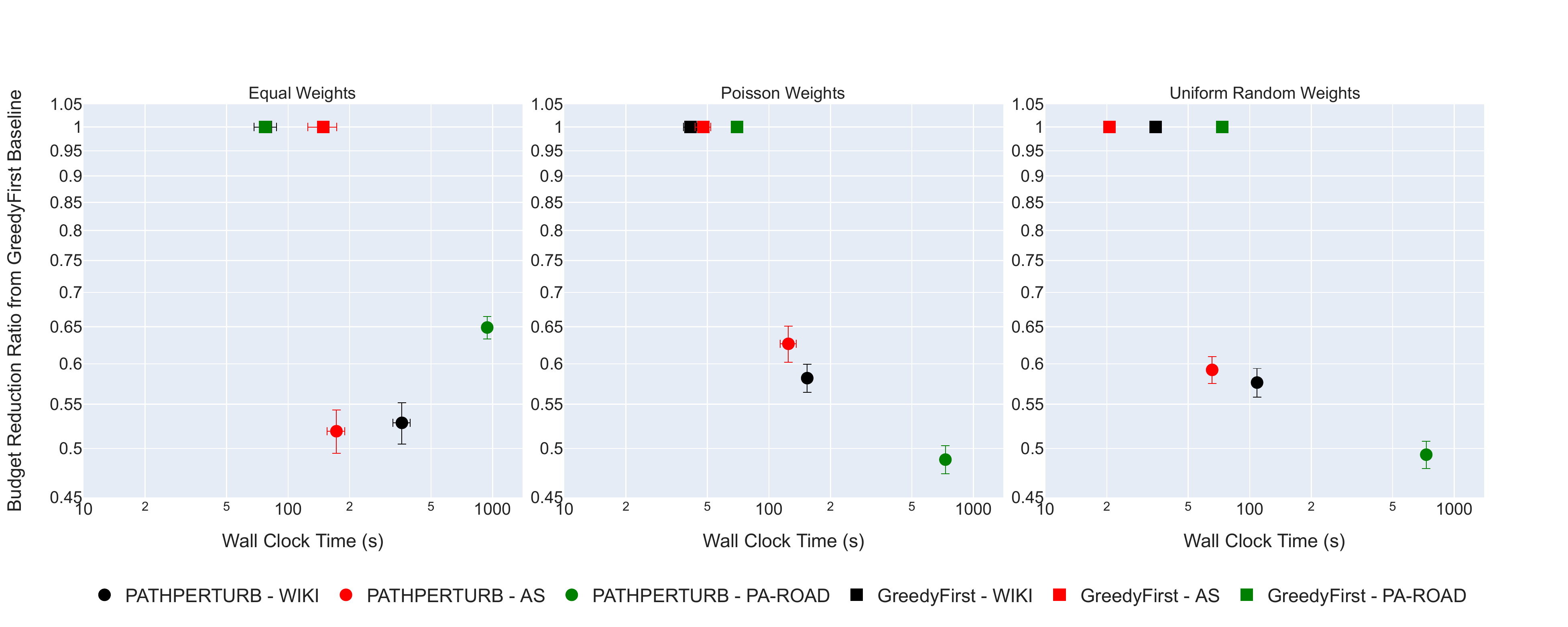}
    \caption{Results on unweighted real networks.
    Results are shown using \texttt{PATHPERTURB} ($\circ$) and \texttt{GreedyFirst} ($\square$), and each color represents a different network. The plots present results when the weights are equal (left), when they are drawn from a Poisson distribution (center), and drawn from a uniform distribution (right). Each plot shows the required budget as a proportion of the budget required using \texttt{GreedyFirst} (vertical axis) with respect to wall clock running time (horizontal axis). Lower cost reduction ratio and lower wall clock time (toward the lower left) is better. 
    Error bars represent standard errors. As with the synthetic networks, \texttt{PATHPERTURB} provides a significant cost reduction in all networks, though in this case we see a greater increase in running time for PA-ROAD. Note: the black square (\texttt{GreedyFirst} on WIKI) in the left-hand plot is obscured by the green square.}
    \label{fig:plots_ruw_gs}
\end{figure*}

\begin{figure*}
    \centering
    \includegraphics[width=0.6\textwidth]{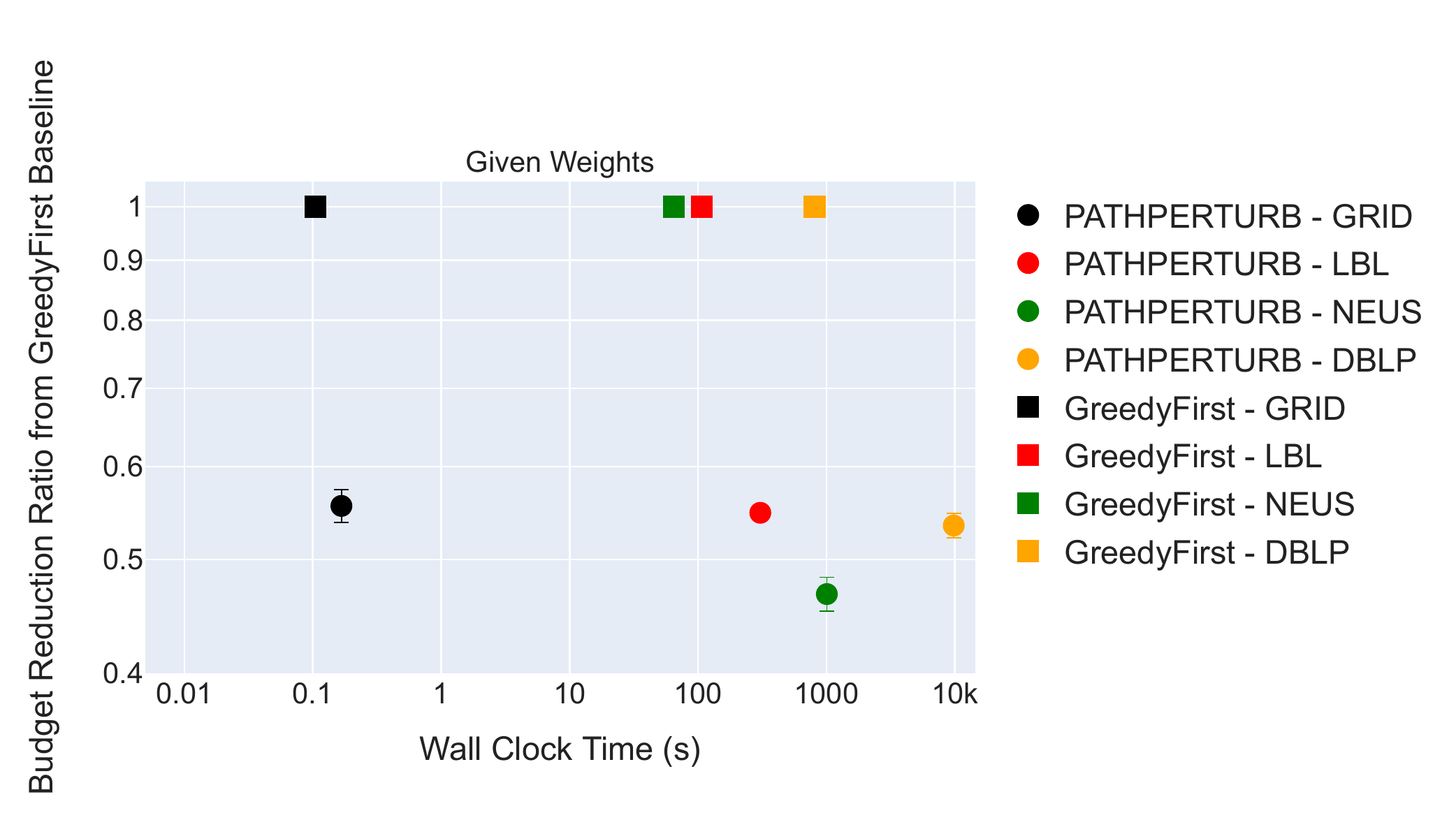}
    \caption{Results on weighted real networks.
    Results are shown using \texttt{PATHPERTURB} ($\circ$) and \texttt{GreedyFirst} ($\square$), and each color represents a different network. The plot shows the required budget as a proportion of the budget required using \texttt{GreedyFirst} (vertical axis) with respect to wall clock running time (horizontal axis). Lower cost reduction ratio and lower wall clock time (toward the lower left) is better. 
    Error bars represent standard errors. In all cases, \texttt{PATHPERTURB} reduces the cost of attacking the graph by about a factor of two.}
    \label{fig:plots_rw_gs}
\end{figure*}

Across all experiments, we see a substantial improvement over \texttt{GreedyFirst} by using \texttt{PATHPERTURB}, for the most part ranging from a 25\% reduction in the required perturbation budget to a decrease of more than a factor of two. This comes at the expense of increased running time: an increase of an order of magnitude appears typical. In the synthetic networks, \texttt{PATHPERTURB} provides a greater improvement for heterogeneous Kronecker (KRON) and BA graphs rather than Erd\H{o}s--R\'{e}nyi graphs. This could be an effect of hubs: nodes with high degree that tend to facilitate short paths may make it more difficult to obtain a low budget via a greedy procedure.

The main exception to the substantial budget improvement is cliques (COMP, blue in Figure~\ref{fig:plots_sim_gs}) with Poisson weights. To understand this result, consider an unweighted clique, and note that the 800th shortest path is a 3-hop path. The optimal perturbation to make a particular 3-hop path the shortest is to perturb all edges adjacent to $s$ and $t$ except those on $p^*$. With weights drawn from a Poisson distribution, we get weights typically near the mean, which gives us a similar effect to the unweighted graph: 2-hop paths are all similar lengths, as are 3-hop paths. In this context, \texttt{GreedyFirst} identifies a near-optimal perturbation. Looking deeper into the data, we see that for path ranks of 100, 200, and 400, the required budgets using \texttt{GreedyFirst} and \texttt{PATHPERTURB} are exactly the same when all edge weights are equal.

We observe one major difference from the results with \texttt{PATHATTACK}~\cite{Miller2021}, where the adversary's goal is the same but his/her attack vector is to cut edges. We see more substantial gains over the greedy baseline in grid-like networks. In the lattice network and the road networks, \texttt{PATHATTACK} took substantially more time for very modest improvements in edge removal cost. Here, we see a relatively high computational burden in these grid-like networks---reliably an order of magnitude in the real datasets---but the budget reductions are among the best. In addition, we note that lattices with Poisson weights are one of the few cases where \texttt{GreedyMin} outperforms \texttt{GreedyFirst} (the other being cliques with Poisson weights), though the difference is small and does not explain the extent of the difference. This may be due to the difference in cost.  In~\cite{Miller2021}, costs were proportional to the weights of removed edges, while in the present work the cost of perturbing a path will be smaller if the edge weights are larger. Thoroughly investigating this phenomenon is a subject for future work.

\section{Related Work}
\label{sec:related}
This paper expands the work on adversarial graph analysis that was introduced recently. Examples include attacks against vertex classification~\cite{Zugner2018,Zugner2019b} and node embedding~\cite{Bojchevski2019}, as well as community detection when an adversary does not want to be grouped with other individuals~\cite{Kegelmeyer2018}. In \cite{Miller2021}, an adversary cut edges to attack shortest path algorithms; here, an adversary adds edge weights. 

Prior network science work on attacks against graphs was focused on attempts to disrupt infrastructure, e.g., disconnecting a power grid graph. In this area, it was shown that graphs with heterogeneous degree distributions (like BA and KR graphs) are much more robust to random node removals, but highly susceptible to targeted attacks against the nodes with the most connections~\cite{Albert2000}. 

There has been previous work on altering shortest paths, though it has been primarily focused on removal of a single edge or node. The objective of the ``most vital edge'' (or most vital node) problem is, given two nodes in a graph, to find the edge (node) whose removal most increases the shortest path between the source and destination~\cite{Nardelli2001,Nardelli2003}. In an adversarial context, this would be an instance of an adversary intending to divert the user from the best solution, rather than being motivated to push traffic along a particular path of interest.


In that sense, the most vital edge and node problem is similar to recent work on Stackelberg planning~\cite{Speicher2018}. In this work, like in the most valuable node and edge problems, the goal of the attacker is to make it as costly as possible for the user to perform the task. While the most valuable edge only allows one move, the Stackelberg planning work uses a turn-based leader-follower framework, where the leader makes the follower's actions more costly at each step. A similar problem is the adversarial shortest path problem, where the state space has uncertainty that an adversary could exploit to decrease the user's reward as states are traversed~\cite{Neu2012}.
%

There are two complementary areas where path finding in an adversarial context is crucial. One is network interdiction, in which an adversary is attempting to traverse a network undetected~\cite{Washburn1995}. The other involves planning paths through hostile territory; for example, an unmanned aerial vehicle in enemy air space~\cite{Jun2003}. Recent path interdiction work has focused on attack disruption~\cite{Letchford2013}. Work in this area also uses oracles to judiciously select from an extremely large set of potential strategies~\cite{Jain2011}.

\section{Conclusions}
\label{sec:conclusion}

We defined the Force Path Problem, in which an adversary adds weights to edges in order to make a particular path the shortest between a pair of source and destination nodes. The adversary has a budget, which he/she cannot exceed. We showed that Force Path can be optimized to within an arbitrarily small error in polynomial time. We demonstrated that Force Path can be formulated as a linear program with an intractable number of constraints. However, standard shortest-path algorithms can be used to obtain a polynomial-time constraint oracle, which allowed us to use constraint generation. We formalized this procedure in our \texttt{PATHPERTURB} algorithm, which we applied to a diverse collection of real and simulated data. We observed that the perturbation budget optimized using \texttt{PATHPERTURB} is often as little as half of what can be obtained using a greedy baseline perturbation procedure.

\section*{Acknowledgments}
This material is based upon work supported by the United States Air Force under Air Force Contract No. FA8702-15-D-0001 and the Combat Capabilities Development Command Army Research Laboratory (under Cooperative Agreement Number W911NF-13-2-0045). Any opinions, findings, conclusions or recommendations expressed in this material are those of the authors and do not necessarily reflect the views of the United States Air Force or Army Research Laboratory.
%
%
%
%
\bibliographystyle{named}
\bibliography{bibfile.bib}


\end{document}